\documentclass{article}

\usepackage[table]{xcolor}

\usepackage{amsthm}
\usepackage{booktabs}
\usepackage{setspace}
\usepackage{braket}
\usepackage{caption}
\usepackage{bm}
\usepackage{mathtools}
\usepackage{algorithm}
\usepackage{algorithmicx,algpseudocode}
\usepackage[normalem]{ulem}
\usepackage{empheq}
\usepackage[framemethod=tikz]{mdframed}

\usepackage[margin=2cm]{geometry}
\usepackage{newtxtext}
\usepackage{newtxmath}

\usepackage[absolute]{textpos}

\mdfdefinestyle{mystyle}{
	hidealllines=true,
	leftline=true,
	innerleftmargin=10pt,
	innerrightmargin=10pt,
	innertopmargin=0pt,
}
\surroundwithmdframed[style=mystyle]{proof}

\makeatletter
\def\thmheadbrackets#1#2#3{%
	\thmname{#1}\thmnumber{\@ifnotempty{#1}{ }\@upn{#2}}%
	\thmnote{ {\the\thm@notefont[#3]}}}
\makeatother

\newtheoremstyle{brakets}% Name
{}% space above
{}% space below
{\itshape}% body font
{}% indent
{\bfseries}% head font
{.}% punctuation after head
{ }% space after head (has to be space or dimension!)
{\thmheadbrackets{#1}{#2}{#3}}% head spec

\newtheoremstyle{defbrakets}% Name
{}% space above
{}% space below
{\normalfont}% body font
{}% indent
{\bfseries}% head font
{.}% punctuation after head
{ }% space after head (has to be space or dimension!)
{\thmheadbrackets{#1}{#2}{#3}}% head spec

\newtheoremstyle{defproblem}% Name
{}% space above
{}% space below
{\normalfont}% body font
{}% indent
{\bfseries}% head font
{.}% punctuation after head
{ }% space after head (has to be space or dimension!)
{\thmheadbrackets{#1}{#2}{#3}}% head spec

\theoremstyle{brakets}
\newtheorem{thm}{Theorem}

\theoremstyle{defbrakets}
\newtheorem{cor}{Corollary}
\newtheorem{defn}{Definition}
\newtheorem{prop}{Proposition}
\theoremstyle{defproblem}
\newtheorem{plm}{Problem}
\newtheorem{rem}{Remark}

\newcommand{\overbar}[1]{\mkern 1.5mu\overline{\mkern-1.5mu#1\mkern-1.5mu}\mkern 1.5mu}

\algnewcommand\algorithmicforeach{\textbf{for each}}
\algdef{S}[FOR]{ForEach}[1]{\algorithmicforeach\ #1\ \algorithmicdo}

\algnewcommand\algorithmicinput{\textbf{Input:}}
\algnewcommand\INPUT{\item[\algorithmicinput]}

\algnewcommand\algorithmicoutput{\textbf{Output:}}
\algnewcommand\OUTPUT{\item[\algorithmicoutput]}

\let\oldReturn\Return
\renewcommand{\Return}{\State\oldReturn}

\makeatletter
\newcommand{\algrule}[1][.2pt]{\par\vskip.5\baselineskip\hrule height #1\par\vskip.5\baselineskip}
\makeatother
\begin{document}

\title{Reconfigurable Timed Discrete-Event Systems}

\author{Matin Macktoobian}
\date{School of Engineering \\
Swiss Federal Institute of Technology in Lausanne (EPFL)\\
Lausanne, Switzerland \\
matin.macktoobian@epfl.ch}

\maketitle

\begin{textblock}{14}(2.5,1)
	\noindent\textbf{\color{red}Published in ``2020 24th International Conference on System Theory, Control and Computing (ICSTCC)''\\ DOI: 10.1109/ICSTCC50638.2020.9259790}
\end{textblock}

\begin{abstract}
In this paper, we present the first general solution to the automatic reconfiguration problem of timed discrete-event systems. We extend the recursive forcible backtracking approach which had been already solved the automatic reconfiguration problem of untimed discrete-event systems. In particular, we first solve the timed centralized reconfiguration problem using a specific timed eligibility set. Then, we study the identity between the solutions to an arbitrary timed centralized reconfiguration problem and its corresponding decentralized version. It turns out that the solutions to both cases are identical to each other. So, the solution obtained by the proposed theory is interestingly invariant to systematic distributions.
\end{abstract}
\textit{keywords}: Timed Backtracking Forcibility, Automatic Reconfiguration, Timed Discrete-Event Systems, Complex Systems
\section{Introduction}
\doublespacing
The efficient control of complex systems often depends on real-time systems requirements. In particular, dynamic reconfiguration is utmost of importance in the real-time control of complex systems. For example, real-time reconfiguration of networked data centers \cite{baccarelli2016energy} were studied to maximize the computing capacity and minimize the data loss in the course of reconfigurations. Similarly, modern manufacturing systems have been reconfigured based on the notion of Internet of Things \cite{bi2014internet}. Power systems also actively sought real-time reconfigurations \cite{rao2013power}. Intelligent approaches were also employed for temporal reconfiguration of microgrids' topological structures \cite{shariatzadeh2014real}. Middle-ware approaches to software reconfiguration \cite{oreizy1998role} and industrial supervisory systems \cite{pfitscher2013intelligent} were the other examples of the complex systems in which timed reconfiguration problem is taken into account.

The literature of real-time reconfiguration indicates that most of the contributions have addressed particular requirements of the timed reconfiguration problem. In other words, each solution addresses a restricted set of the requirements of real-time reconfiguration based on its intended application. Thus, existence of an all-in-one general solution which covers all (or at least majority of) those capabilities is an open problem. To propose such a relatively general solution, we first investigate the literature to identify the functionalities which are expected to exist in a real-time reconfigurable system. Subsequently, we introduce our approach to solve the reconfiguration problem of timed discrete-event systems (TDES) into which the required features are embedded.

Optimized reconfiguration \cite{vrba2010capabilities} is a fundamental requirement of complex systems with many components. One can optimize a particular reconfiguration temporally and/or spatially. Spatial optimization refers to obtaining the shortest operational path from one configuration to another in the case of untimed systems. Our method is spatially optimized since the system can choose the forcible path solving the problem with the shortest length. So, the desired untimed reconfiguration is performed by the minimum of operational steps. On the other hand, temporal optimization is directly related to timed reconfiguration scenarios. We later illustrate that our strategy optimizes the time corresponding to a reconfiguration by executing the fastest timed forcible path solving the intended timed reconfiguration.

Another optimization-related issue of real-time reconfiguration is the complexity of online computations corresponding to real-time optimizations. Namely, some reconfiguration techniques, e.g., \cite{zeller2012timing}, use optimizers. These optimizers satisfy a class of (in)equalities constrained to particular state of a system. Specifically, these optimizers minimize many parameters including the list of the required components to which the system has to migrate in the course of a reconfiguration. On the other hand, our scheme efficiently finds the optimal solution which is the forcible path with the minimum number of ticks, so the optimization is resolved without complex constraints.

TDES and supervisory control theory have also been used in some real-time reconfiguration applications in a limited manner, e.g., in task execution on processors. In particular, real-time scheduling on uni-processors were studied in \cite{wang2016dynamic}. As another example, \cite{chen2002real} investigated the non-preemptive execution supervision of periodic tasks on a processor in which interruptions are not allowed prior to the completion of task executions. These applications use traditional behavioral specifications of supervisory control theory to model reconfigurations. Expressed differently, the synthesized supervisors exclusively manage the behavioral reconfigurations. By contrast, our methodology implements both behavioral specifications and reconfiguration specification, so not only are the behavioral requirements of the system fulfilled, but also the planned reconfigurations are realized.

In this paper, we reformulate the untimed reconfiguration problem (\cite{macktoobian2017automatic,macktoobian2018automatic}) to establish the timed reconfiguration problem. Subsequently, we solve the problem using a timed recursive backtracking approach which yields the timed forcible paths solving the timed reconfiguration problem. Specifically, we synthesize a timed reconfiguration supervisor corresponding to a particular timed reconfiguration problem with respect to its reconfiguration specification. Then, we propose a dynamic-programming-based \cite{bellman1966dynamic} timed backtracking forcibility strategy to automatically solve the timed reconfiguration problem applicable to a variety of finite-state-automaton-based TDES. We apply temporal optimization to the solutions to timed reconfiguration problems to achieve the fastest solutions. We also partition the reconfiguration time associated with a particular timed reconfiguration problem to its constituents. Thus, the timing profiles of those constituents can be customized  according to any arbitrary reconfiguration requirement. As well, we solve the timed decentralized reconfiguration problem, then we clarify the relationships between an arbitrary timed centralized reconfiguration problem and its decentralized variant. We assert that this strategy is the first general solution to the automatic reconfiguration problem of TDES. 

This paper is structured as follows. Section \ref{sec:preliminaries} briefly introduces TDES. We solve the timed centralized reconfiguration problem in Section \ref{sec:cent}. In particular, Section \ref{subsec:RS-TCRS} demonstrates the overall approach to synthesizing the timed centralized reconfiguration supervisor (TCRS) into which the solutions to its associated timed reconfiguration problem are embedded. We define the timed reconfiguration problem and the the notion of timed backtracking forcibility in Sections \ref{subsubsec:problem} and \ref{subsubsec:TBFC}, respectively. Then, Section \ref{sec:decent} solves the timed decentralized reconfiguration problem. Section \ref{sec:example} solves the illustrated example in \cite{macktoobian2017automatic} in the timed case. We outline our accomplishments corresponding to the automatic reconfiguration of TDES in Section \ref{sec:conclusion}.
\vspace*{3mm}
\section{Preliminaries}
\label{sec:preliminaries}
Brandin-Wonham framework \cite{brandin1994supervisory} defined TDES by adjoining time bounds to transitions of untimed DES (UDES). Particularly, plant \textbf{G} starts from an untimed activity transition graph (ATG) $\textbf{G}_{\text{act}} := (A, \Sigma_{\text{act}}, \delta_{\text{act}}, a_{0}, A_{m})$ with $\Sigma := \Sigma_{\text{act}} \dot{\cup} \{tick\}$. Each $a \in A$ denotes an \textit{activity}. We have $\Sigma_{\text{act}} := \Sigma_{\text{spe}} \dot{\cup} \Sigma_{\text{rem}}$, where $\Sigma_{\text{spe}}$ (respectively, $\Sigma_{\text{rem}}$) is the \textit{prospective} (respectively, \textit{remote}) event set with \textit{finite} (respectively, \textit{infinite}) upper time-bounds \cite{wonham2017supervisory}. We define the \textit{timer interval} $T_{\sigma}$ for event $\sigma$ to be $[0,l_{\sigma}]$ and $[0,u_{\sigma}]$ for $\sigma \in \Sigma_{\text{rem}}$ and $\sigma \in \Sigma_{\text{spe}}$, respectively. Then, the \textit{initial state} is $q_{0} := (a_{0}, \{t_{\sigma_{0}}|\sigma \in \Sigma_{\text{act}}\})$, where $t_{\sigma_{0}}$ is $l_{\sigma}$ or $u_{\sigma}$ for a prospective or remote event, respectively. The \textit{marker state set} is $Q_{m} \subseteq A_{m} \times \prod\{T_{\sigma}|\sigma \in \Sigma_{\text{act}}\}$. Thus, a TDES is represented by $\textbf{G} := (Q, \Sigma, \delta, q_{0}, Q_{m})$. An event $\sigma \in \Sigma_{\text{act}}$ is \textit{eligible} at $q$, i.e., $\delta_{\text{act}}(a,\sigma)!$, if $\delta_{\text{act}}(a,\sigma)$ is defined; it is \textit{eligible}, i.e., $\delta(q,\sigma)!$, if $\delta(q,\sigma)$ is defined. Additionally, the \textit{closed behavior} and the \textit{marked behavior} of \textbf{G} are the languages
\begin{subequations}
	\begin{empheq}[left={}]{align}
		L(\textbf{G}) &:= \{s\in \Sigma^{*}|\delta(q_0,s)!\},
		\label{eq:1}\\
		L_{m}(\textbf{G}) &:= \{s\in L(\textbf{G})|\delta(q_0,s) \in Q_{m}\}.
		\label{eq:2}
	\end{empheq}
\end{subequations}
\textbf{G} is \textit{nonblocking} if $\overbar{L_{m}(\textbf{G})} = L(\textbf{G})$, where $\overbar{L_{m}(\textbf{G})}$ is the \textit{prefix closure} of $L_{m}(\textbf{G})$. The \textit{eligible event set} ${\text{Elig}}_{\textbf{G}}(s) \subseteq \Sigma$ at $q$ corresponding to $s \in L(\textbf{G})$ is associated with TDES \textbf{G} is ${\text{Elig}}_{\textbf{G}}(s) := \{\sigma \in \Sigma|s\sigma \in L(\textbf{G})\}$. Given arbitrary language $K \subseteq L(\textbf{G})$ and $s \in \overbar{K}$, ${\text{Elig}}_{K}(s) := \{\sigma \in \Sigma|s\sigma \in \overbar{K}\}$. Considering the set of all controllable sublanguages of $K$, denoted by $\mathcal{C}(K)$, $\text{sup}\mathcal{C}(K)$ represents its supremal element. Given specification $E \subseteq \Sigma^{*}$, there exists a \textit{monolithic supervisor} \textbf{S} such that
\begin{subequations}
	\begin{empheq}[left={}]{align}
		L_{m}(\textbf{S}) &:= \text{sup}\mathcal{C}(E \cap L_{m}(\textbf{G})),\label{eq:3}\\
		L(\textbf{S}) &:= \overbar{L_{m}(\textbf{S})}.
		\label{eq:4}
	\end{empheq}
\end{subequations}
\section{Centralized Reconfiguration of TDES}
\label{sec:cent}
\subsection{Timed Centralized Reconfiguration Supervisor synthesis}
\label{subsec:RS-TCRS}
The definitions of reconfiguration event and reconfiguration specification associated with timed reconfiguration problem exactly resemble the the definitions corresponding to the untimed centralized reconfiguration problem \cite{macktoobian2018automatic}.

To compute the TCRS associated with a particular timed reconfiguration problem, we start by composing all activity transition graphs (ATGs) corresponding to $n$ components of TDES \textbf{G} and reconfiguration specification \textbf{R} leading to the multimodal version of \textbf{G}, say \textbf{GMode}, as follows\footnote{A complete elaboration on the supervisory functions used in this section is presented in \cite{wonham2017supervisory}.}.
\begin{equation}
	\textbf{GMode}_{\text{ATG}} := \uline{\textbf{compose}}(\textbf{G}_{1,\text{ATG}}, \cdots, \textbf{G}_{n,\text{ATG}},\textbf{R})
\end{equation}
Then, the timed transition graph (TTG) associated with $\textbf{GMode}_{\text{ATG}}$ is computed as the following:
\begin{equation}
	\textbf{GMode}_{\text{TTG}} := \uline{\textbf{timed\_graph}}(\textbf{GMode}_{\text{ATG}}).
\end{equation}	
Next, the global system specification is defined by composing behavioral specification \textbf{E} and all events of $\textbf{GMode}_{\text{TTG}}$. Finally, the \textbf{TCRS} is computed as the following:
\begin{equation}
	\textbf{TCRS} := \uline{\textbf{supcon}} (\textbf{GMode}_{\text{TTG}}, [\uline{\textbf{allevents}}(\textbf{GMode}_{\text{TTG}})\parallel \textbf{E}]).
\end{equation}
We complete the design process by presenting an algorithm to solve the timed reconfiguration problem. We construct a backtracking algorithm that collects all timed forcible paths reaching a suitable target state of a \textbf{TCRS} from an arbitrary source state. To hit this mark, we first present the formal definition of the timed reconfiguration problem, then obtain the conditions for timed backtracking forcibility, and finally propose the algorithm solving the problem.
\subsection{Timed Backtracking Forcibility and Solvability Checking}
\label{subsec:solvability}
\subsubsection{Problem Statement}
\label{subsubsec:problem}
Similar to the untimed reconfiguration problem, we intend to trigger a desired reconfiguration upon its request at any state of a \textbf{TCRS}. A reconfigurable TDES takes a timed reconfiguration into account when the reconfiguration event associated with the timed reconfiguration occurs. Thus, we have to find a path (set) whose event(s) can be forced to occur, reaching a state at which the desired reconfiguration event is defined. The solvability of a timed reconfiguration problem is guaranteed if at least one forcible path can be found to activate the intended timed reconfiguration. Thus, the timed reconfiguration problem is defined as follows.
\begin{plm}[Timed Reconfiguration Problem]
	Denote by $q_s$ the state where \textbf{TCRS} currently resides, and let $q_r$ be a state at which a desired RE, say $\sigma_r$, is defined; namely 
	$\delta_\textbf{TCRS}(q_r,\sigma_r)!$  Subject to an appropriate specification of forcibility, determine the timed forcible path (set) from $q_s$ to $q_r$.
\end{plm}
\subsubsection{Timed backtracking Forcibility and Solvability Checking}
\label{subsubsec:TBFC}
In this section, we update the definitions corresponding to the recursive backtracking algorithm presented in \cite{macktoobian2018automatic}. So, the approach is efficiently applicable to the timed case.

First, the algorithm needs to keeps track of the timed forcible paths in the course of traversing the states of the timed centralized reconfiguration supervisor \textbf{TCRS}. So, Given the current state $q \in Q_{\textbf{TCRS}}$ corresponding to a timed backtracking problem, we define the notion of \textit{timed eligibility set} $\Lambda_q$ (with respect to $q$) as follows.
\begin{equation}
	\begin{split}
		\label{eq:eligT}
		\Lambda_{q} := \bigl\{ (q',\sigma) \mid &
		\overbrace{q' \in Q_{\textbf{TCRS}} \wedge \sigma \in \Sigma_{\textbf{TCRS}} \wedge \delta_{\textbf{TCRS}}(q',\sigma)=q}^{I}
		\wedge \\ & \hspace*{-2.7cm}\underbrace{[(\sigma \in \Sigma_{\mathrm{for}}) \vee (\forall \sigma' \in \Sigma_{\textbf{TCRS}}) (\delta_{\textbf{TCRS}}(q',\sigma')! \neq q \Rightarrow ( \sigma' \in \Sigma_{\mathrm{hib}}))]}_{II}
		\bigr\}
	\end{split}
\end{equation}
Here, $\Sigma_{\textbf{TCRS}} := \Sigma_{\text{for}} \cup \Sigma_{\text{hib}} \dot{\cup} \{tick\}$.

According to the subformula I, $q$ has to be one-step reachable from $q'$. The subformula II defines the required disablement or forcibility conditions for a successful backtracking to $q'$. In particular, if $\sigma$ is forcible, then it preempts uncontrollable event $\sigma'$ eligible at $q'$. $\Omega_{q}$ determines the state-event tuples corresponding to the forcible transitions to $q$ in the backtracking process.

The following definition extracts the first element of each $\Lambda_q$'s tuple, i.e., $q'$.
\begin{defn}[Selector Function]
	Let \textbf{TCRS} be a timed centralized reconfiguration supervisor. Given $q \in Q_{\textbf{TCRS}}$, let $\Lambda_q$ be an eligibility set including tuples $(q',\sigma)$ in which backtracking $q$ to $q' \in Q_{\textbf{TCRS}}$ is eligible via $\sigma \in \Sigma_{\textbf{TCRS}}$. Then, selector function $P_{1}: \Omega_{q} \rightarrow Q_{\textbf{TCRS}}$ is defined as follows
	\begin{equation}
		P_{1}(\Lambda_q) := \{q'|(\exists \sigma \in \Sigma_{\textbf{TCRS}}) (q',\sigma) \in \Lambda_q\}.
	\end{equation}
\end{defn}
Now, the definitions of backtracking forcibility tree (BFT) and proper backtracking forcibility tree (P-BFT) are adapted to the timed case as follows.
\begin{defn}[Backtracking Forcibility Tree]
	Assume a timed centralized reconfiguration problem associated with timed centralized reconfiguration supervisor \textbf{TCRS}. Let also $q_{s} \in Q_\textbf{TCRS}$ and $q_{r} \in Q_\textbf{TCRS}$ be the source state and the target state corresponding to the problem, respectively. Then, the \textit{backtracking forcibility tree} (BFT) corresponding to the problem is a tree whose root element is $q_r$ and the remainder of its nodes and links are recursively generated by the following node-link generator:
	\begin{equation}
		\label{eq:rbftT}
		\mathcal{T}(q,\Lambda_{q}) := 
		\begin{cases}
			\mathcal{T}(q,\emptyset) &\hspace*{-3mm}\text{if } q=q_{s} \text{          (terminal case)},\\
			{\bigcup}_{q' \in P_{1}(\Lambda_q)}\mathcal{T}(q',\Lambda_{q'}) &\hspace*{-3mm}\text{if } q \neq q_{s} \text{          (inductive case)}.\\
		\end{cases}
	\end{equation}
\end{defn}
\begin{defn}[Proper Backtracking Forcibility Tree]
	Given a timed centralized reconfiguration problem with respect to a timed centralized reconfiguration supervisor \textbf{TCRS}, a target state $q_r \in Q_{\textbf{TCRS}}$, and a source state $q_s \in Q_{\textbf{TCRS}}$, the BFT corresponding to the problem is \textit{proper} if all of its leaves are $q_s$.
\end{defn}
Using the definition of attraction field (see, \cite{macktoobian2018automatic}), we extend the notions of forcible path, branching path, and direct path to their timed counterparts as below.
\begin{defn}[Timed Forcible Path (Set)]
	\label{defn:tfp}
	Let \textbf{TCRS} be a timed centralized reconfiguration supervisor associated with a particular timed reconfiguration problem with source state $q_{s} \in \Sigma_{\textbf{TCRS}}$ and target state $q_{r} \in \Sigma_{\textbf{TCRS}}$. Given the attraction field $\bm{Z}$ corresponding to the problem, the \textit{set of timed forcible paths} solving the problem is defined as follows.
	\begin{equation}
		\mathcal{P} := \{\pi \in \Sigma_{\textbf{TCRS}}^{*} | \overbrace{\delta_{\textbf{TCRS}}(q_{s},\pi) = q_{r}}^{I} \wedge \overbrace{\delta_{\textbf{TCRS}}(q_{s},\overline{\pi}) \subseteq \bm{Z}}^{II} \}
	\end{equation}
\end{defn}
\begin{defn}[Branching and Direct Paths]
	Let \textbf{TCRS} be a timed centralized reconfiguration supervisor associated with a particular timed reconfiguration problem with source state $q_{s} \in \Sigma_{\textbf{TCRS}}$ and target state $q_{r} \in \Sigma_{\textbf{TCRS}}$. Let $\omega \in \Sigma_{\textbf{TCRS}}^{*}$ be a timed forcible path, say, $\delta_{\textbf{TCRS}}(q_{s},\omega) = q_{r}$, i.e., $q_r$ is reachable from $q_s$. Then, given any forcible path $\pi \in \Sigma_{\textbf{TCRS}}^{*}$ such that $\pi \neq \omega$ and $\delta_{\textbf{TCRS}}(q_{s},\pi) = q_{r}$, $\pi$ is a \textit{branching path} and $\omega$ is a \textit{direct path} with respect to $\pi$ corresponding to the problem.
\end{defn}
Finally, we use the $\mathcal{FTP}$ algorithm \cite{macktoobian2018automatic} to define the timed reconfiguration solver ($\mathcal{TRS}$) algorithm 
\begin{equation}
	\label{eq:trscT}
	\mathcal{TRS}: Q_{\textbf{TCRS}} \times Q_{\textbf{TCRS}} \times \Sigma_{\textbf{TCRS}}^{*} \rightarrow \Sigma_{\textbf{TCRS}}^{*},
\end{equation}
as Algorithm \ref{alg:TRSC}.
\begin{algorithm}[t]
	\setstretch{1.3}
	\caption{Timed Reconfiguration Solver ($\mathcal{TRS}$)}
	\label{alg:TRSC}
	\begin{algorithmic}[1] 
		\INPUT 
		\Statex $q_r$ \Comment Target state
		\Statex $q_s$ \Comment Source state
		\Statex \textbf{TCRS} \Comment Timed centralized reconfiguration supervisor
		\OUTPUT
		\Statex $\mathcal{P}$ \Comment Timed forcible path set
		\algrule[1pt] 
		\State compute the BFT associated with \textbf{TCRS}, $q_r$, and $q_s$
		\State P-BFT $\leftarrow \mathcal{FTP}(\text{BFT})$ \Comment prune the BFT
		\State read the direct paths from the P-BFT
		\State include the branching paths (if any exists)
		\Return $\mathcal{P}$
	\end{algorithmic}
\end{algorithm}

Similar to the $\mathcal{URS}$ algorithm \cite{macktoobian2018automatic}, the $\mathcal{TRS}$ algorithm yields a non-empty timed forcible path set $\mathcal{P}$ if the desired timed reconfiguration problem is solvable. The proof of the correctness and the computational complexity of the $\mathcal{TRS}$ algorithm resemble those of the $\mathcal{URS}$ algorithm \cite{macktoobian2017automatic}.

As well, our backtracking process is invariant to the presence or the absence of tick. In view of the backtracking process, the tick projection of the timed centralized reconfiguration supervisor only shrinks the event set $\Sigma_{\textbf{TCRS}} := \Sigma_{\text{for}} \cup \Sigma_{\text{hib}} \dot{\cup} \{ tick\}$ to $\Sigma_{\textbf{TCRS}} := \Sigma_{\text{for}} \cup \Sigma_{\text{hib}}$, and this evolution changes neither the definition of the eligibility set nor the functionality of our backtracking process.
\subsubsection{Functional Compositional Computation of Tick-Projected Timed Forcible Paths}
\label{subsubsec:func-comm1}
A timed forcible path, which generally includes ticks, in fact represents a spatio-temporal dynamics in that the timed forcible path's tick substrings exhibit its temporal evolution, but its other events render the spatial (or operational) aspects of the solution. In this part, we are interested in the result of applying the tick-projection operator to a TCRS before and after solving the timed reconfiguration problem associated with it using the $\mathcal{TRS}$ algorithm. To this end, Proposition \ref{prop:commm} shows that the result will be the spatial projection associated with the intended timed forcible path regardless of the order of applying the $\mathcal{TRS}$ algorithm and the tick-projection operator to the timed centralized reconfiguration supervisor.

We recall that $\textit{project}: \Sigma_{\textbf{TCRS}}^{*} \rightarrow (\Sigma_{\textbf{TCRS}}\setminus\{tick\})^{*}$ projects out the ticks of a path (set). Additionally the signature of the $\mathcal{TRS}$ algorithm (see, (\ref{eq:trscT})) is $\mathcal{TRS}: Q_{\textbf{TCRS}} \times Q_{\textbf{TCRS}} \times \Sigma_{\textbf{TCRS}}^{*} \rightarrow \Sigma_{\textbf{TCRS}}^{*}$. We note that $(\Sigma_{\textbf{TCRS}}\setminus\{tick\})^{*} \subseteq \Sigma_{\textbf{TCRS}}^{*}$, so the \textit{project} and the $\mathcal{TRS}$ algorithm can be functionally composed in any order.
\begin{prop}
	\label{prop:commm}
	Let the strings corresponding to a timed forcible path (set) $\mathcal{P}$ be a solution (set) to a particular timed centralized reconfiguration problem associated with the timed centralized reconfiguration supervisor \textbf{TCRS} and source and target states $q_{s},q_{r} \in Q_{\textbf{TCRS}}$ with respect to  reconfiguration event $\sigma_r$. Given P$\mathcal{P}$ as the spatial projection of $\mathcal{P}$, the following equations hold.
	\begin{equation}
		\label{eq:projTRSC}
		\begin{split}
			&\textit{project}(\mathcal{TRS}(q_{s},q_{r},L(\textbf{TCRS})))  = \\ &\mathcal{TRS}(q_{s},q_{r},\textit{project}(L(\textbf{TCRS}))) = \text{P}\mathcal{P}
		\end{split}
	\end{equation}
\end{prop}
\begin{proof}
	Let P\textbf{TCRS} be the tick-projected version of \textbf{TCRS}. We need to show that both the functional compositions reach a state in P\textbf{TCRS} at which $\sigma_r$ is eligible to occur. Note that the operator $\text{P}(\cdot)$ is substituted for the function \textit{project} in the proof for brevity of notation.
	\leavevmode
	\begin{itemize}
		\item $project \circ \mathcal{TRS}$:\\
		Let $\pi \in \mathcal{P}$. By the definition of timed forcible path, we have $\pi \cdot \sigma_{r} \in L(\textbf{TCRS})$. By applying tick projection, we obtain $\text{P}(\pi \cdot \sigma_{r}) \in L(\text{P}\textbf{TCRS})$, i.e., $\text{P}(\pi)\cdot \text{P}(\sigma_{r}) \in L(\text{P}\textbf{TCRS})$, thereby $\text{P}(\pi)\cdot \sigma_{r} \in L(\text{P}\textbf{TCRS})$. Considering $q \in Q_{\text{P}\textbf{TCRS}}$ to be the current state of P\textbf{TCRS}, we have $(\exists q' \in Q_{\text{P}\textbf{TCRS}}) \delta_{\text{P}\textbf{TCRS}}(q,\text{P}(\pi)\cdot \sigma_{r})=q'$. Since $\delta_{\text{P}\textbf{TCRS}}(q,\text{P}(\pi)\cdot \sigma_{r})!$, we can conclude that $\delta_{\text{P}\textbf{TCRS}}(\delta_{\text{P}\textbf{TCRS}}(q,\text{P}(\pi)), \sigma_{r})!$
		
		\item $\mathcal{TRS} \circ project$:\\ 
		Let $\pi \in \mathcal{P}$. Since $\pi \in L(\textbf{TCRS})$, $\text{P}(\pi) \in L(\text{P}\textbf{TCRS})$, thereby $\delta_{\text{P}\textbf{TCRS}}(q,\text{P}(\pi))!$. Note that $\pi \cdot \sigma_{r} \in L(\textbf{TCRS})$ and since $\pi$ solves the timed reconfiguration problem in \textbf{TCRS}, then 
		$\text{P}(\pi \cdot \sigma_{r}) \in L(\text{P}\textbf{TCRS})$, so we have $\text{P}(\pi) \cdot \sigma_{r} \in L(\text{P}\textbf{TCRS})$. This clearly implies that $\delta_{\text{P}\textbf{TCRS}}(q,\text{P}(\pi)\cdot \sigma_{r})!$, that is, $\delta_{\text{P}\textbf{TCRS}}(\delta_{\text{P}\textbf{TCRS}}(q,\text{P}(\pi)), \sigma_{r})!$.
	\end{itemize} 
	Since applying both of the functional compositions to \textbf{TCRS} yields P$\mathcal{P}$, we conclude that the commutativity holds. 
\end{proof}
\begin{rem}
	\label{rem:efficient}
	We note that $\mathcal{TRS} \circ \textit{project}$ is computationally more efficient than $\textit{project} \circ \mathcal{TRS}$ since the former first eliminates the ticks of \textbf{TCRS}, so the $\mathcal{TRS}$ algorithm has to process a simpler supervisor with fewer number of transitions (and probably fewer states).
\end{rem}
\vspace*{4mm}
\section{Decentralized Reconfiguration of TDES}
\label{sec:decent}
We find that a timed reconfiguration problem can be solved by a particular timed forcible path (set) in both centralized and decentralized manners.
\begin{defn}[Decentralization Package]
	Let \textbf{G} be a TDES controlled by timed centralized reconfiguration supervisor \textbf{TCRS}. Given an event list EV based on which the localization is defined, the 3-tuple $\Delta := (\textbf{G},\textbf{TCRS}, \text{EV})$ is a \textit{decentralization package} with respect to \textbf{G}.
\end{defn}
\begin{rem}
	We denote the second element of $\Delta$ by $\Delta \bigg\rvert_{\textbf{TCRS}} := \textbf{TCRS}$.
\end{rem}
The following theorem formulates the solution equivalence between timed centralized and decentralized cases by a functional compositional law.
\begin{thm}
	\label{thm:comm2}
	Let $\Delta := (\textbf{G},\textbf{TCRS}, \text{EV})$ be a decentralization package with respect to the TDES \textbf{G}. Let $\mathcal{P}$ be the strings corresponding to the timed forcible path (set)  solving an arbitrary timed reconfiguration problem with respect to \textbf{TCRS}. Given source and target states $q_{s},q_{r} \in Q_{\textbf{TCRS}}$ associated with reconfiguration event $\sigma_r$, $\mathcal{P}$ also solves the timed reconfiguration problem with respect to $\textbf{TDRS} := \uline{\textbf{\text{timed}\_\text{localize}}}(\Delta)$. In other words, the following equalities hold hold for some $q',q'' \in Q_{\textbf{TDRS}}$.
	\begin{equation}
		\label{eq:comm2}
		\begin{split}
			\mathcal{TRS}(q_{s},q_{r},L(\Delta \bigg\rvert_{\textbf{TCRS}})) &=\\ \mathcal{TRS}(q',q'',\uline{\textbf{timed\_localize}} (\Delta)) &= \mathcal{P}
		\end{split}
	\end{equation}
\end{thm}
\begin{proof}
	According to the timed supervisor localization theory, $L(\textbf{G}) \cap L(\textbf{TDRS}) = L(\textbf{TCRS})$, thus $L(\textbf{TCRS}) \subseteq L(\textbf{TDRS})$. Since $\mathcal{P} \subseteq L(\textbf{TCRS})$, we conclude that $\mathcal{P} \subseteq L(\textbf{TDRS})$. Let $Q_{\textbf{TCRS}}$ and $\delta_{\textbf{TCRS}}(\cdot,\cdot)$ be the state set and the transition function corresponding to \textbf{TCRS}, respectively. Hence, considering $q \in Q_{\textbf{TCRS}}$ as the currently occupied state of \textbf{TCRS}, $\mathcal{P}$ guarantees that $(\forall \pi \in \mathcal{P})\delta_{\textbf{TCRS}}(\delta_{\textbf{TCRS}}(q,\pi),\sigma_r)!$. Let also $Q_{\textbf{TDRS}}$ and $\delta_{\textbf{TDRS}}(\cdot,\cdot)$ be the state set and the transition function corresponding to \textbf{TDRS}, respectively. Let $\pi \in \mathcal{P}$ be the string corresponding to an arbitrary timed forcible path. Suppose $q' \in Q_{\textbf{TDRS}}$ is the current state of \textbf{TDRS}, and $\pi$ reaches $q'' \in Q_{\textbf{TDRS}}$ from $q'$, i.e., $\delta_{\textbf{TDRS}}(q',\pi) = q''$. We have to show that $\delta_{\textbf{TDRS}}(q'',\sigma_r)!$.
	
	By definition, we have $\textbf{TDRS}: = \textbf{LOC}^{P} \parallel \textbf{LOC}^{C}$, where $\textbf{LOC}^{P}$ and $\textbf{LOC}^{C}$ are the  global localized tick controller and the global localized event controller corresponding to \textbf{TDRS}, respectively. Let $\Sigma_{\textbf{LOC}^{C}_{\alpha}}$ and $\Sigma_{\textbf{LOC}^{P}_{\beta}}$ be the event sets of the localized event controller and the localized tick controller corresponding to events $\alpha$ and $\beta$, respectively; then event sets of $\textbf{LOC}^{P}$ and $\textbf{LOC}^{C}$ can be expressed as follows
	\begin{subequations}
		\begin{empheq}[left={\empheqlbrace\,}]{align}
			&\Sigma_{\textbf{LOC}^{P}} :=
			\bigcup_{\beta \in \Sigma_{\text{for}}} \Sigma_{\textbf{LOC}^{P}_{\beta}}, \\ &\Sigma_{\textbf{LOC}^{C}} :=
			\bigcup_{\alpha \in \Sigma_{\text{hib}}} \Sigma_{\textbf{LOC}^{C}_{\alpha}},
		\end{empheq}
	\end{subequations}
	\noindent where $\Sigma_{\textbf{TDRS}} := \Sigma_{\textbf{LOC}^{P}} \cup \Sigma_{\textbf{LOC}^{C}}$. Consequently, given the following inverse projection operators
	\begin{subequations}
		\begin{empheq}[left={\empheqlbrace\,}]{align}
			&\text{P}_{P}^{-1} :=
			\text{pwr}(\Sigma_{{\textbf{LOC}}^{P}}^{*}) \longrightarrow  \text{pwr}(\Sigma^{*}_{\textbf{TDRS}}),\\ &\text{P}_{C}^{-1} :=
			\text{pwr}(\Sigma_{{\textbf{LOC}}^{C}}^{*}) \longrightarrow  \text{pwr}(\Sigma^{*}_{\textbf{TDRS}}),
		\end{empheq}
	\end{subequations}
	since $\pi\cdot\sigma_r \in L(\textbf{TCRS})$, thereby $\pi\cdot\sigma_r \in L(\textbf{TDRS})$, i.e., $\pi\cdot\sigma_r \in L(\textbf{LOC}^{P} \parallel \textbf{LOC}^{C})$ implying that
	\begin{subequations}
		\begin{empheq}[left={\empheqlbrace\,}]{align}
			&\pi\cdot\sigma_r \in \text{P}_{P}^{-1}(\textbf{LOC}^{P}), \label{third:1}\\
			&\pi\cdot\sigma_r \in \text{P}_{C}^{-1}(\textbf{LOC}^{C}).\label{third:2}
		\end{empheq}
	\end{subequations}
	Let $q_P \in Q_{\textbf{LOC}^{P}}$ and $q_C \in Q_{\textbf{LOC}^{C}}$ be the states at which $\textbf{LOC}^{P}$ and $\textbf{LOC}^{C}$ currently reside, respectively; thus, we have
	\begin{subequations}
		\begin{empheq}[left={\empheqlbrace\,}]{align}
			&(\exists q_{C} \in Q_{\textbf{LOC}^{P}}) \delta_{\textbf{LOC}^{P}}(q_P, \text{P}_{P}(\pi\cdot\sigma_r)) = q_{C},\\
			&(\exists q_{P} \in Q_{\textbf{LOC}^{C}}) \delta_{\textbf{LOC}^{C}}(q_C, \text{P}_{C}(\pi\cdot\sigma_r)) = q_{P}.
		\end{empheq}
	\end{subequations}
	We observe that $\sigma_{r} \in \Sigma_{\textbf{LOC}^{C}}$, since $\sigma_{r} \in \Sigma_{\text{hib}}$ and it has to be controllable to be enabled whenever a reconfiguration is desired. Moreover, note that $\sigma_{r} \in \Sigma_{\textbf{LOC}^{P}}$, since it has to be able to preempt its competing events that are eligible to occur at the target state to initialize the desired reconfiguration successfully. Thus, we have $\sigma_{r} \in \Sigma_{\textbf{LOC}^{P}} \cap \Sigma_{\textbf{LOC}^{C}}$, i.e., 
	\begin{subequations}
		\begin{empheq}[left={\empheqlbrace\,}]{align}
			&(\exists q_{C} \in Q_{\textbf{LOC}^{P}}) \delta_{\textbf{LOC}^{P}}(q_P, \text{P}_{P}(\pi)\cdot\sigma_r) = q_{C},\\
			&(\exists q_{P} \in Q_{\textbf{LOC}^{C}}) \delta_{\textbf{LOC}^{C}}(q_C, \text{P}_{C}(\pi)\cdot\sigma_r) = q_{P}.
		\end{empheq}
	\end{subequations}
	Therefore, we conclude that
	\begin{subequations}
		\begin{empheq}[left={\empheqlbrace\,}]{align}
			& \delta_{\textbf{LOC}^{P}}(\delta_{\textbf{LOC}^{P}}(q_P, \text{P}_{P}(\pi)),\sigma_r)!,\\
			& \delta_{\textbf{LOC}^{C}}(\delta_{\textbf{LOC}^{C}}(q_C, \text{P}_{C}(\pi)),\sigma_r)!.
		\end{empheq}
	\end{subequations}
	That is, any arbitrary timed forcible path $\pi$ reaches the states of both $\textbf{LOC}^{P}$ and $\textbf{LOC}^{C}$ at which $\sigma_r$ is eligible to occur, and $\mathcal{P}$ overall solves the timed reconfiguration problem with respect to \textbf{TDRS} as well.
\end{proof}
\begin{rem}
	The theorem above in fact demonstrates that our strategy indeed processes both centralized and decentralized reconfiguration problem in the same manner; this point again clarifies the fact that the solutions to the two problems are identical.
	
	We also observe that the timed forcible path set corresponding to \textbf{TDRS} exists in both $\textbf{LOC}^{P}$ and $\textbf{LOC}^{C}$ of \textbf{TDRS}. Specifically, we assert that both $\textbf{LOC}^{P}$ and $\textbf{LOC}^{C}$ belong to the domain space of the $\mathcal{TRS}$ algorithm since $\Sigma_{\textbf{LOC}^{P}}^{*} \subseteq \Sigma_{\textbf{TCRS}}^{*}$ and $\Sigma_{\textbf{LOC}^{C}}^{*} \subseteq \Sigma_{\textbf{TCRS}}^{*}$.
\end{rem}
We note that the timed decentralized reconfiguration problem is also invariant to the order of applying the $\mathcal{TRS}$ algorithm and the tick projection. In particular, the result is the spatial projection of the timed forcible path as stated by the following corollary.
\begin{cor}
	\label{cor:commmD}
	Let $\mathcal{P}$ be the set of strings corresponding to a timed forcible path (set) which solves to a timed decentralized reconfiguration problem associated with a timed decentralized reconfiguration supervisor \textbf{TDRS}. Suppose also source and target states $q_{s},q_{r} \in Q_{\textbf{TDRS}}$ are defined with respect to reconfiguration event $\sigma_r$. Given P$\mathcal{P}$ as the spatial projection of $\mathcal{P}$, the following equality holds.
	\begin{equation}
		\label{eq:projTDRS}
		\begin{split}
			&\textit{project}(\mathcal{TRS}(q_{s},q_{r},L(\textbf{TDRS})))  =  \\ & \mathcal{TRS}(q_{s},q_{r},\textit{project}(L(\textbf{TDRS}))) = \text{P}\mathcal{P}
		\end{split}
	\end{equation}
\end{cor}
\begin{proof}
	The claim is proved by substituting \textbf{TCRS} with \textbf{TDRS} in the proof of Proposition \ref{prop:commm}.
\end{proof}
\section{Example}
\label{sec:example}
In this section, we solve a timed reconfiguration variation of SMALL FACTORY, whose untimed version was solved in \cite{macktoobian2017automatic}. We have to compute the timed centralized reconfiguration supervisor associated with a desired timed reconfiguration scenario, so the $\mathcal{TRS}$ algorithm can check the solvability of the problem. The time bounds corresponding to the events of SMALL FACTORY are specified in Table. \ref{tbl:events}.
\begin{table}[htp]
	\setlength\belowcaptionskip{-0.7\baselineskip}
	\centering
	\caption{The timing characteristics of the manufacturing cell's \newline\hspace*{1.1cm} events}
	\label{tbl:events} 
	\begin{tabular}{ccc}
		\toprule
		Event label & Lower bound & Upper bound\\
		\midrule
		11 & 1 & $\infty$\\
		12 & 0 & 3\\
		13 & 1 & $\infty$\\
		20 & 1 & 2\\
		22 & 0 & 4\\
		23 & 1 & $\infty$\\
		30 & 2 & 4\\
		31 & 2 & $\infty$\\
		32 & 2 & 4\\
		33 & 2 & $\infty$\\
		\bottomrule
	\end{tabular}
\end{table}
Additionally, we take $\braket{13}$, $\braket{23}$, $\braket{31}$, and $\braket{33}$ into account as forcible events. The reconfiguration specification \textbf{R} is, also, planned including reconfiguration event $\braket{91}$ which is a remote event with lower bound 2. Then, the reconfigurable model of the plant is obtained by composing all agents of SMALL FACTORY and \textbf{R}  as the following:
\begin{equation}
	\textbf{RMACH} := \uline{\textbf{compose}}(\textbf{M1},\textbf{M2},\textbf{R}).
\end{equation}
Then, the timed transition graph corresponding to \textbf{RMACH}, i.e., \textbf{TRMACH}, is computed.
\begin{equation}
	\begin{split}
		\textbf{TRMACH} &:= \uline{\textbf{timed\_graph}}(\textbf{RMACH}) \\
		\textbf{ALLTRMACH} &:= \uline{\textbf{allevents}}(\textbf{TRMACH})
	\end{split}
\end{equation}
According to the specification described in \cite{macktoobian2018automatic}, we synthesize the \textbf{TSUP} supervisor, which is fully controllable.
\begin{equation}
	\textbf{TSUP} =: \uline{\textbf{supcon}}(\textbf{TRMACH},\textbf{SPEC})
\end{equation}
Assume that \textbf{TSUP} resides at $[\bm{62}]$ and a reconfiguration task is triggered, so a user selects $[\bm{0}]$ at which  $\braket{91}$ is eligible to occur. Thus, we need to check whether or not backtracking from $[\bm{0}]$ to $[\bm{62}]$ is possible considering our recursive backtracking approach. Consequently, the $\mathcal{TRS}$ algorithm backtracks from $[\bm{0}]$ to $[\bm{62}]$ to find all timed forcible paths. The shortest path is $\pi_{1} := \braket{23,33,tick,12,31}$.

We present evidence for the results of Proposition \ref{prop:commm} Specifically, we eliminate ticks of \textbf{TSUP} as follows.
\begin{equation}
	\textbf{PTSUP} = \uline{\textbf{project}}(\textbf{TSUP},Null[0])
\end{equation}
First we observe that $\text{P}(\pi_{1}) = \braket{23,33,12,31}$ reaches a state of P\textbf{TSUP} at which $\braket{91}$ is defined:
\begin{equation}
	\text{P}(\pi_{1}): [\bm{4}] \xrightarrow{\braket{23}} [\bm{7}] \xrightarrow{\braket{33}} [\bm{17}] \xrightarrow{\braket{12}} [\bm{20}] \xrightarrow{\braket{31}} [\bm{0}] \checkmark
\end{equation}
Second, applying timed backtracking forcibility to the timed reconfiguration problem corresponding to P\textbf{TSUP} illustrates that the backtracking from $[\bm{0}]$ to $[\bm{4}]$ is authorized according to $\text{P}(\pi)$. Thus, the $\mathcal{TRS}$ algorithm and \textit{project} can be applied to a timed centralized reconfiguration supervisor in either order to obtain the spatial projections of a particular timed forcible path (set). We also observe that the state size and the number of events of P\textbf{TSUP} are much smaller than the state and event size of \textbf{TSUP}. This implies that the processing of P\textbf{TSUP} is computationally simpler then \textbf{TSUP} for the $\mathcal{TRS}$ algorithm. Thus, the validity of the claim of Remark \ref{rem:efficient} is confirmed in this instance.
\vspace*{4mm}
\section{Conclusion}
\label{sec:conclusion}
Temporal considerations always strengthen the complication of the management of complex systems. TDES theory synchronized by supervisory control theory provides a powerful structure to model and to control real-time complex systems. This research solved the automatic reconfiguration problem of TDES. We found the solutions to the timed reconfiguration problem which are the timed forcible paths triggering intended reconfigurations in reconfiguration supervisors. Interestingly, both centralized and decentralized versions of the timed reconfiguration problem are solved by identical solutions.
\vspace*{4mm}
\section*{Acknowledgement}
This work was supported by the Natural Sciences and Engineering Research Council (NSERC), Grant Number DG-480599, in the period in the course of which the author was with the Systems Control Group, the University of Toronto, ON, Canada. The author also appreciates helpful comments from W. M. Wonham.
\vspace*{4mm}
\bibliographystyle{IEEEtran}
\bibliography{references}{}
\end{document}